\documentclass{article}

\usepackage{arxiv}

\usepackage{hyperref}       
\usepackage{url}            
\usepackage{booktabs}       
\usepackage{amsfonts}       
\usepackage{nicefrac}       
\usepackage{microtype}      
\usepackage{comment}
\usepackage[]{algorithm2e}
\usepackage{amsmath,amssymb,amsthm,xfrac,xcolor,bm}
\usepackage{tikz}    
\usetikzlibrary{positioning}
\usepackage[utf8]{inputenc} 
\usepackage[T1]{fontenc}    

\usepackage{microtype}      

\usepackage{amsmath,amssymb,xfrac,xcolor,bm,comment}

\usepackage{algorithm2e}

\renewcommand{\S}{Section~}

\newcommand{\Uopt}[2]{{U^*_{#1}(#2)}}
\newcommand{\hUopt}[2]{{\hat U^*_{#1}(#2)}}
\newcommand{\Upi}[2]{{U^\pi_{#1}(#2)}}
\newcommand{\hUpi}[2]{{\hat U^\pi_{#1}(#2)}}
\newcommand{\Ugrid}[2]{{U^*_{#1,\eta}(#2)}}
\newcommand{\Ulen}[3]{{u^*_{#1}(#2,#3)}}
\newcommand{\Ulenpi}[3]{{u^\pi_{#1}(#2,#3)}}
\newcommand{\Ugridlen}[3]{{u^*_{#1,\eta}(#2,#3)}}
\newcommand{\scalarPrice}{{\mu}}
\newcommand{\perr}{\underline{\varepsilon}}
\newcommand{\hp}{\hat{\bm p}^\ell_{t,s}}
\newcommand{\hP}{\hat{\bm P}(\ell)}
\renewcommand{\L}{\mathcal L}
\newcommand{\V}{\mathcal V}
\newcommand{\D}{\mathcal D}
\newcommand{\LL}{\mathbb L}
\newcommand{\VV}{\mathbb V}
\newcommand{\DD}{\mathbb D}
\newcommand{\KK}{\mathbb K}
\renewcommand{\SS}{\mathcal S}
\newcommand{\R}{\mathbb R}
\newcommand{\rev}{\mathsf{Rev}}
\newcommand{\crev}{\mathsf{CmlRev}}
\newcommand{\ecrev}{\mathsf{ECRev}}
\DeclareMathOperator{\Expectation}{\mathbb E}
\newcommand{\E}[2]{\Expectation\limits_{#1}\left[#2\right]}
\newcommand{\X}{\mathcal X}

\newcommand{\F}{\bar F}
\newcommand{\ceil}[1]{\left\lceil #1\right\rceil}
\newcommand{\floor}[1]{\left\lfloor #1\right\rfloor}

\newtheorem{theorem}{Theorem}
\newtheorem{claim}[theorem]{Claim}
\newtheorem{lemma}[theorem]{Lemma}
\newtheorem{corollary}[theorem]{Corollary}
\newtheorem{definition}[theorem]{Definition}

\title{Online Revenue Maximization for Server Pricing}

\author{
  Shant Boodaghians\thanks{Partially supported by NSF grant CCF-1750436.} \\
  University of Illinois at Urbana-Champaign,\\
  Urbana IL 61801, USA \\
  \texttt{boodagh2@illinois.edu} \\
   \And
  Federico Fusco\thanks{Supported by ERC Advanced Grant 788893 AMDROMA ``Algorithmic and Mechanism Design Research in Online Market'' and MIUR PRIN project ALGADIMAR ``Algorithms, Games, and Digital Markets''} \\
  Department of Computer, Control\\ and Management Engineering\\
  Sapienza University\\
  Rome, Italy \\
  \texttt{fuscof@diag.uniroma1.it}
   \And
   Stefano Leonardi\footnotemark[2] \\
  Department of Computer, Control\\ and Management Engineering\\
  Sapienza University\\
  Rome, Italy \\
  \texttt{leonardi@diag.uniroma1.it}
   \And
  Yishay Mansour \\
  Tel Aviv University,\\
  P.O. Box 39040, Tel Aviv 6997801, Israel\\
  \texttt{mansour@tau.ac.il}
  \And
  Ruta Mehta\footnotemark[1] \\
  University of Illinois at Urbana-Champaign,\\
  Urbana IL 61801, USA \\
  \texttt{rutamehta@cs.illinois.edu} \\
}

\begin{document}

\maketitle

 \begin{abstract}
    Efficient and truthful mechanisms to price resources on remote servers/machines has been the subject of much work in recent years due to the importance of the cloud market.
    This paper considers revenue maximization in the online stochastic setting with non-preemptive jobs and a unit capacity server. 
    One agent/job arrives at every time step, with parameters drawn from an underlying distribution.
    We design a posted-price mechanism  
    which can be efficiently computed, and is revenue-optimal in expectation and in retrospect, up to additive error. 
    The prices are posted prior to learning the agent's type, and
    the computed pricing scheme is deterministic, depending only on the length of the allotted time interval and on the earliest time the server is available.  
    We also prove that the proposed pricing strategy is robust to  imprecise knowledge of the job distribution and that a distribution learned from polynomially many samples is sufficient to obtain a near-optimal truthful pricing strategy. 
\end{abstract}

\clearpage

\section{Introduction}
Designing mechanisms for a desired outcome with
strategic and selfish agents is an extensively studied problem in economics, 
with classical work by
Myerson \cite{myerson81}, and Vickrey-Clarke-Groves \cite{vickrey61}, for truthful mechanisms. 
The advent of online interaction and e-commerce has added an efficiency constraint on the mechanisms, going so far as to prioritize computational efficiency over classical objectives: {\em e.g.} choosing simple
approximate mechanisms when optimal mechanisms are computationally difficult, or impossible. 
Beginning with
Nisan and Ronen \cite{NR99}, the theoretical computer science community has contributed greatly to the field,
in both fundamental problems and specific applications.
These include designing truthful mechanisms for the maximization of welfare and revenue, and has also focused on learning distributions of agent types, menu complexity, and dynamic mechanisms ({\em e.g.}, \cite{denBoer2015dynamic,cole-roughgarden-2014-sample}.) 

We consider this question in the setting of selling 
computational resources on remote servers or machines ({\em cf.} \cite{tang-li-fu2017budget,Babaioff2017}.)  This is arguably one of the fastest growing markets on the Internet.  
The goods (resources)
are assigned non-preemptively and thus 
have strong complementarities. 
Furthermore, since the supply (server capacity) is limited, any mechanism trades immediate revenue for future supply. Finally, mechanisms must be incentive-compatible, 
as non-truthful, strategic, behaviour from the agents can skew the performance of a mechanism from its theoretical guarantees. This leads us to the following question:

\subparagraph*{Question.} {Can we design an efficient, truthful, and revenue-maximizing mechanism to sell time-slots non-preemptively on a single server?}

\smallskip

We design a posted-price mechanism which maximizes expected revenue up to additive error, for agents/buyers arriving online, with parameters of value, length and maximum delay, drawn from an underlying distribution. 
 
Three key aspects distinguish our problem
from standard online scheduling: 
(i)~In~our setting, as time progresses, the server clears up, allowing longer jobs to be scheduled in the future if no smaller jobs are scheduled until then. 
(ii)~Scheduling the jobs is not exclusively to the discretion of the mechanism designer, but must also be desired by the job itself, while also producing sufficient revenue. 
(iii)~As~the mechanism designer, we do not have access to job parameters in an incentive-compatible way before deciding on a posted price menu.
These three features lie at the core of the difficulty of our problem. 
Our focus will be on devising online mechanisms
in the Bayesian setting. 

In our online model, time on the server is discrete. 
At every time step, an agent arrives on the server, with a value $V$, length requirement $L$, and maximum delay $D$.
These parameters are drawn from a common distribution, {\em i.i.d.} across jobs.
The job wishes to be scheduled for at least $L$ consecutive time slots, no more than $D$ time units after its arrival, and wishes to pay no more than $V$. 
Jobs are assumed to have quasi-linear utility in money, and so prefer the least-price interval within their constraints.
The mechanism designer never learns the parameters of the job.
Instead, she posts a price menu of (length,price) pairs, and the minimum available delay $s$. The job accepts to be scheduled so long as $D\geq s$, and there is some (length,price) pair in the menu of length at least $L$ and price at most $V$.
We note that the pricing scheme can be dynamic, changing through time. 
If, at time epoch $t$, an agent chooses option $(\ell,\pi_\ell)$, then she pays $\pi_\ell$ and her job will be allocated to the interval $[t+s,t+s+\ell]$. She will choose the option which minimizes $\pi_\ell$.
Throughout this paper we assume that the random variables $L,V,D$ are discrete, and have finite support, unless specified differently.

\subsection{Summary of Our Results}

\begin{enumerate}
\item  We model the problem of finding a revenue maximizing pricing strategy as a Markov Decision Process (MDP).  
Given a price menu (length,price) and a state (minimum available delay) $s$ at time $t$,  
the probability of transition to any other state at time $t+1$ is obtained from the distribution of the job's parameters.  
The revenue maximizing pricing strategy can be efficiently computed via backwards induction.  
We also present, in Appendix~\ref{sec:Approx-Grid-Values}, an approximation scheme in the case where $V$ is a continuous random variable.

\item We prove that the optimal pricing strategy is monotone in length under a distributional assumption, which
we show is satisfied when the jobs' valuation follows a log-concave distribution, parametrized by length. 
Recall that log-concave distributions are exactly those which have a monotone hazard rate.
This implies the existence of an optimal pricing mechanism which ensures truthfulness in the finite horizon setting when the distributions are known. In Appendix~\ref{sec:Infinite-Horizon-Approx}, this is extended to the infinite discounted horizon setting,
incurring a small additive error.
We also demonstrate good concentration bounds of the revenue obtained by the optimal truthful posted price strategy. 

\item We finally investigate the robustness of the pricing strategy.  We first show that a near optimal solution is still obtained when the distribution is known with a certain degree of uncertainty. We complement this result by analyzing the performances of the proposed pricing strategy when the distribution is only known from samples collected through the observations of the agents' decisions.
We provide a truthful posted price  $\varepsilon$-approximate mechanism if the number of samples is polynomial in $1/\varepsilon$ and the size of the support of the distribution. 
\end{enumerate}

\subsection{Related Work}
Much recent work has focused on 
designing efficient mechanisms for pricing cloud resources. 
Chawla et al. \cite{chawla17} recently studied ``time-of-use'' pricing mechanisms, to match demand to supply with deadlines and online arrivals.
Their result assumes large-capacity servers, and seeks to maximize welfare in a setting in which the jobs arriving over time are not i.i.d..
\cite{Azar2015TruthfulOS} provides a mechanism 
for preemptive scheduling with deadlines, 
maximizing the total value of completed jobs. 
Another possible objective for the design of incentive-compatible scheduling mechanisms is the total value of completed jobs,
which have release times and deadlines.
\cite{Porter:2004:MDO:988772.988783} solves this problem in an online setting, while \cite{DBLP:conf/icpp/CarrollG08}, in the offline setting for parallel machines, and \cite{DBLP:conf/atal/StrohleGWSR14}, in the online competitive setting with uncertain supply.   
\cite{Jain2013ATM} focuses on social welfare maximization for non-preemptive scheduling on multiple servers, and obtains a constant competitive ratio as the number of servers increases.
Our work differs from these by considering revenue maximization and stochastic job types which are i.i.d. over time.
\cite{Kilcioglu2016CompetitionOP} addresses computing a price menu for revenue maximization with
different machines.
Finally, \cite{Babaioff2017} proposes a system architecture for scheduling and pricing in cloud computing. 

Posted price mechanisms (PPM)  have been introduced by \cite{sandholm2004} and have gained attention due to their simplicity, robustness to collusion, and their ease of implementation in practice. 
One of the first theoretical results concerning PPM's is an asymptotic comparison to classical single-parameter mechanisms
\cite{bh08}. 
They were later studied by \cite{chms10} for the objective of revenue maximization, and further strengthened by \cite{kw12} and \cite{dk15}. 
\cite{fgl15} shows that sequential PPM's can $\sfrac{1}{2}$-approximate social welfare
for XOS valuation functions, if the price for an item is equal to the expected contribution of the item to the social welfare.

Sample complexity for revenue maximization was recently been studied in \cite{cole-roughgarden-2014-sample} showing that a polynomially many of samples suffice to obtain near optimal Bayesian auction mechanisms. An approach based on statistical learning that allows to learn mechanisms with expected revenue arbitrarily close to optimal from a polynomial number of samples has been proposed in \cite{Morgenstern:2015:PNA:2969239.2969255}.  The problems of learning simple auctions from samples has been studied in \cite{pmlr-v49-morgenstern16}.

\subsection{Structure of the Paper}
In~\S\ref{sec:model} we describe the model of the problem as a Markov Decision Process. 
In~\S\ref{sec:Bayes-opt-policies} we present an efficient algorithm for computing optimal policies for the finite time horizon given full knowledge of the distribution of the jobs' paramethers.
This is extended to other settings in Appendix~\ref{app:extensions}.
In~\S\ref{sec:Monotone-Opt}, we demonstrate that the optimal policy is monotone and in ~\S\ref{sec:concentration-full-knowledge} we describe the concentration bounds on the revenue of a pricing policy.
~\S\ref{subsec:learning-hat-Q} gives the learning algorithm and error bounds for computing the pricing policies with only (partial)  sample access to the job distribution. 
In 
Finally, \S\ref{sec:performance} and \S\ref{sec:conclusion} are devoted to describing and summarizing the final result and future directions of research. 

Proof details are provided in Appendix~\ref{app:proofs}.

\section{Model}
\label{sec:model}

\subparagraph*{Notation.}
In what follows, the variables $t$, $\ell$ or $L$, $v$ or $V$, and $d$ or $D$ 
are reserved for describing the parameters of a job that wishes to be scheduled.
Respectively, they represent the arrival time $t$, required length $\ell$, value $v$, and maximum allowed delay $d$.
The lowercase variables represent fixed values, whereas the uppercase represent random variables.
Script-uppercase letters $\L,\V,\D$ represent the supports of the distributions on $L$, $V$, 
and $D$, respectively; and the bold-uppercase letters $\LL,\VV,\DD$ represent the maximum values 
in these respective sets.
Finally, $\pi$ is reserved for pricing policy, whereas $p$ is reserved for probabilities.

\subparagraph*{Single-Machine, Non-Preemptive, Job Scheduling.}
A sequence of random jobs wish to be scheduled on a server, non-preemptively, for a sufficiently low price, within a time constraint. 
Formally, at every time step $t$, a single job with parameters $(L,V,D)$ is drawn from an underlying distribution $Q$ over the space $\L\times\V\times\D$.
It wishes to be scheduled for a price $\pi\leq V$ in an interval $[a,b]$ such that $a-t\leq D$ and $b-a\geq L$.

\subparagraph*{Price Menus.}\label{sec:online-model}
Our goal is to design a take-it-or-leave-it, posted-price mechanism which maximizes expected revenue.
At each time period, the mechanism posts a ``price menu'' and an earliest-available-time $s_t$, indicating that times $t$ through $t + s_t - 1$ have already been scheduled. ($s_t$ will henceforth be referred to as the {\em state} of the server.)	
We let $\SS := \{0,\,\dotsc,\,\DD+\LL\}$ to be the set of all possible states.
The state of the server at a given time $t$ is naturally a random variable which depends on the earlier jobs and on the adopted policy $\pi$. 
As before, we will denote with $s$ or $s_t$ the fixed value, and with $S$ or $S_t$ the corresponding random variable.
The price menu will be given by the function $\pi: [T] \times \SS \times \L \to \R$, i.e., if we are a time $t$ and the server is in state $s_t$, then the prices are set according to $\pi_t(s_t,\cdot): \L \to \R.$ 
The reported pair $(\pi_t(s_t,\cdot),s_t)$ is computed by the scheduler's strategy, which we determine in this paper.
Once this is posted, a job $(L, V,D)$ is then sampled {\em i.i.d.} from the underlying distribution $Q$. 

If $V\geq \pi_t(s_t,\ell)$ for some $\ell\geq L$, and $D\geq s_t$, then the job accepts the schedule, and reports the length $\ell\geq L$ which minimize price.
Otherwise, the job reports $\ell=0$ and is not scheduled.
To guarantee truthfulness, it suffices to have $\pi_t(s,\cdot)$ be monotonically non-decreasing for every state $s$: the agent would not want a longer interval since it costs more, and would not want one of the shorter intervals since they cannot run the job.
It should be clear that the mechanism's strategy is to always report monotone non-decreasing prices, as a decrease in the price menu will only cause more utilization of the server, without accruing more revenue. The main technical challenge in this paper, then, is to show that under some assumptions, the optimal strategy is monotone non-decreasing, and efficiently computable.

\subparagraph*{Revenue Objective.}
Revenue can be measured in either a {\em finite} or an {\em infinite discounted} horizon. 
In the former (finite) case, only $T$ time periods will occur, 
and we seek to maximize the expected sum of revenue over these periods.
In the infinite-horizon setting, future revenue is discounted, at an exponentially decaying rate. Formally, revenue at time $t$ is worth a $\gamma^t$ fraction of revenue at time 0, for some fixed $\gamma<1$. See Appendix~\ref{sec:Infinite-Horizon-Approx}.
Recall that the job parameters are drawn independently at random from the underlying distribution, so the scheduler can only base their ``price menu'' on the state of the system and the current time. 
Thus, the only realistic strategy is to fix a state-and-time-dependent pricing policy $\pi : [T]\times \SS \times \L \to \R$, ``$\pi_t(s,\ell)$'', where $[T]:=\{0,\,1,\,\dotsc,\,T\}$.

Let $\mathcal X=\{\X_1:=(1,L_1,V_1,D_1),\,\X_2:=(2,L_2,V_2,D_2),\,\X_3,\,\dotsc\}$ be the random sequence of jobs arriving, sampled {\em i.i.d.}
from the underlying distribution.    
Let $\pi:[T]\times \SS\times \L\to \R$ be the pricing policy.
We denote as $\rev_t(\mathcal X,\pi)$ the revenue earned at time $t$ with policy $\pi$ and sequence $\mathcal X$.
If $\X_t$ does not buy, then $\rev_t(\mathcal X,\pi)=0$, and otherwise, it is equal to $\pi_t(s_t,L_t)$.
We denote as $\crev_T$
the total (cumulative) revenue earned over the $T$ periods.
Thus,
\begin{equation}
\crev_T(\mathcal X,\pi):=\textstyle\sum_{t=0}^T \rev_t(\mathcal X,\pi).
\end{equation}
We will also need the expected-future-revenue, given a current time and server state, which we will denote as follows:
\begin{equation}
    \label{eq:expected_gain}
\Upi t s=\mathbb{E}_{\X _\geq t}\left[\textstyle\sum_{i=t}^T
    \rev_i(\pi,\X)\;\middle| S_t=s \right],
\end{equation}
The subscript of the expectation 
$\X_{\geq t}$ denotes that we consider only jobs arriving from time $t$ onward.
Our objective is to find the pricing policy $\pi$ which maximizes $\Upi{0}{s=0}$.
Call this $\pi^*$, and denote the expected revenue under $\pi^*$ as $\Uopt{t}{\cdot}$.

\section{Bayes-optimal Strategies for Sever Pricing}
\label{sec:Bayes-opt-policies}

In this section we seek to compute an optimal monotone pricing policy $\pi:[T] \times \SS\times \L \to \R$ which maximizes revenue in expectation over $T$ jobs sampled {\em i.i.d.} from an underlying known distribution $Q$. 
This is extended to the infinite-horizon, discounted, setting in Appendix~\ref{sec:Infinite-Horizon-Approx}.

We first model the problem of maximizing the revenue in online server pricing as a Markov Decision Process that admits an efficiently-computable, optimal pricing strategy. 
The main contribution of this section is to show that, for a natural assumption on the distribution $Q$, the optimal policy is monotone.
We recall that this allows us to derive truthful Bayes-optimal mechanisms.

\subsection{Markov Decision Processes.}
\label{sec:MDP-description}

We show that the theory of {\em Markov Decision Processes} is well suited to model our problem.  A Markov Decision Process is, in its essence, a Markov Chain whose transition probabilities depend on the {\em action} chosen at each state, and where to each transition is assigned a reward. 
A {\em policy} is then a function $\pi$ mapping states to actions.
In our setting, the states are the states of the system outlined in Section~\ref{sec:online-model} (i.e., the possible delays before the earliest available time on the server), and the actions are the ``price menus.''
At every state $s$, a job of a random length arrives, and with some probability, chooses to be scheduled, given the choice of prices. 
The next state is either $\max\{s-1,0\}$, if the job does not choose to be scheduled (since we have moved forward in time), or $s+\ell-1$, if a job of length $\ell$ is scheduled, since we have occupied $\ell$ more units.
The transition probabilities depend on the distribution of job lengths, and the probability that a job accepts to be scheduled given the pricing policy (action). 
Formally, 
\begin{align}
        \mathbb{P}[s_{t+1}=s_t+\ell-1]
        = \begin{cases}
            \mathbb{P}\left[
                L_t=\ell,
                V_t\geq \pi_t(s_t,\ell),
                D_t\geq s_t+\ell
            \right] & \text { if }\ell\geq 1\\
            1-\sum_{k\geq 0}\mathbb{P}[s_{t+1}=s_t+k] & \text{ if }\ell = 0
        \end{cases}\label{eq:MDP-transition-probs}
    \end{align}
{\em (Transitions to state ``$-1$'' should be read as transitions to state ``$\,0\!$''.)}~ 
Note~that a job of length~$\ell$ may choose to purchase an interval of length greater than $\ell$, which would render these transition probabilities incorrect. 
However, this may only happen if the larger interval is more affordable.
It is therefore in the scheduler's interest to guarantee that $\pi_t(s,\cdot)$ in monotone non-decreasing in~$\ell$, which incentivizes truthfulness, since this increases the amount of server-time available, without affecting revenue.
Thus we restrict ourselves to this case.   
    
    It remains to define the transition rewards. They are simply the revenue earned. Formally, 
    a transition from state $s_t$ to $s_t+\ell-1$ incurs a reward of $\pi_t(s,\ell)$, 
    whereas a transition from state $s_t$ to $s_{t}-1$ incurs 0 reward.
We wish to compute a policy $\pi$ in such a way as to maximize the expected cumulative revenue, 
given as the (possibly discounted) sum of all transition rewards in expectation.

\subsection{Solving for the Optimal Policy with Distributional Knowledge}
In this section, we present a modified MDP whose optimal policies can be efficiently computed, and show that these policies are optimal for the original MDP.
In this section, we assume that the mechanism designer is given access to the underlying distribution $Q$. 
However, in the following sections, we will show that if the distribution $Q$ is estimated from samples, then solving for the MDP on this estimated distribution is sufficient to ensure sufficiently good revenue guarantees. 

Since the problem has been modelled as a Markov Decision Process (MDP), we may rely on the wealth of literature available on MDP solutions, in particular we will leverage the {\em backwards induction} algorithm (BIA) of~\cite{puterman2005markov} \S4.5, included in Appendix~\ref{app:proofs} as Algorithm~\ref{alg:MDPAlgorithm}. 
We will however need to ensure that this standard algorithm (i) runs efficiently, and (ii) returns a monotone pricing policy.

Note that past prices do not contribute to future revenue insofar as the current state remains unchanged.
Thus, to compute optimal current prices, we need only know the current state and expected future revenue.
This allows us to use the BIA.
The idea is to compute the optimal time-dependent policy, 
and the incurred expected reward, for shorter horizons, 
then use this to recursively compute the optimal policies for longer horizons.

The total runtime of the BIA is $O(T|\mathcal S||\mathcal A|)$, where $\mathcal S$ and $\mathcal A$ denote the action and state spaces, respectively. 
Note that the dependence on $T$ is unavoidable, since any optimal policy must be time-dependent. 
Recall that $\LL$ and $\DD$ denote the maximum values that $L$ and $D$ can take, respectively, and $\V$ is the set of possible values that $V$ can take.
Denote $\KK:=\max\{\DD+\LL,|\V|\}$. 
If we define the action space na\"ively, we have $|\mathcal S|=\DD+\LL\leq \KK$, and $|\mathcal A|\leq\KK^\LL$.
Thus, a na\"ive definition of the MDP bounds the runtime at $\KK^{O(\KK)}$, which is far from efficient.
Requiring monotonocity only affects lower-order terms.

\subparagraph*{Modified MDP.} To avoid this exponential dependence, we can be a little more clever about the definition of the state space:
instead of states being the possible server states, we define our state space as possible (state, length) pairs. Thus, when the MDP is in state $(s,\ell)$, the server is in state $s$, and a job of length $\ell$ has been sampled from the distribution. 
Our action-space then is simply the possible values of $\pi_t(s,\ell)$, and
the transition probabilities and rewards become:
\begin{align}
	\mathbb{P}[(s,\ell)&\to(s',\ell')|\pi]=\begin{cases}
		\mathbb{P}[V\geq \pi_t(s,\ell),D\geq s|L=\ell]\mathbb{P}[L'=\ell']&\text{ if }s'=s+\ell-1\\
		\mathbb{P}[V<\pi_t(s,\ell) \text{ or } D<s|L=\ell]\mathbb{P}[L'=\ell']&\text{ if }s'=s-1\\
		0&\text{otherwise}
	\end{cases}\label{eq:mod-mdp-transition}\\
	R((s,\ell)&\to(s',\ell')|\pi)=\begin{cases}
		\pi_t(s,\ell)&\text{ if }s'=s+\ell-1\\
		0&\text{otherwise}
	\end{cases}\label{eq:mod-mdp-reward}
\end{align}

Therefore, we get $|\mathcal S|=(\DD+\LL)\cdot \LL\leq \KK^2$, and $|\mathcal A|\leq\KK$.
Thus, the runtime of the algorithm becomes $O(T\KK^3)$.
A full description of the procedure is given in Appendix~\ref{app:proofs} as Algorithm~\ref{alg:finite_horizon}. 
It remains to prove that it is correct.
We begin by claiming that these two MDPs are equivalent in the following sense:
\begin{lemma} 
\label{lemma_equivalence}
For any fixed pricing policy $\pi:[T]\times\SS\times\L\to \R$,\vspace*{-0.35em}
\[
\Upi t s=\E{L}{\Ulenpi t s L}, \forall t \in T , \ s  \in \SS,\\[-0.35em]
\] where the $\Upi t {\cdot}$'s are as in \eqref{eq:expected_gain}, and the $\Ulenpi t {\cdot}{\cdot}$'s are from the modified MDP. 
\end{lemma}
(See Appendix~\ref{app:proofs} for a proof.)
This lemma, however, does not suffice on its own, as agents may behave strategically by over-reporting their length, if the prices are not increasing. 
This would alter the transition probabilities, breaking the analysis. 
We will see that under a mild assumption, this can not happen, as the optimal policy for non-strategic agents will  be monotone, and therefore truthful.

\subsection{Monotonicity of the Optimal Pricing Policies} \label{sec:Monotone-Opt}

Recall that the solution of the more efficient MDP formulation is only correct if we can show that it is always monotone without considering the strategic behaviour of agents, ensuring incentive-compatibility of the optimum.  

An optimal monotone strategy cannot be obtained for all the distributions on $L,V,$ and $D$.  As an example, for any distribution where a job's value is a deterministic function of their length, the optimal policy is to price-discriminate by length. If this function is not monotone, the optimum won't be either.
To this end, we introduce the following assumption, which we will discuss below, and which will imply monotonicity of the pricing policy.

 \subparagraph*{Assumption 1.} \textit{The quantity
\(
\tfrac{\mathbb{P}[V\geq \scalarPrice',D\geq s|L=\ell]}{\mathbb{P}[V\geq \scalarPrice,D\geq s|L=\ell]}
\) is monotone non-decreasing as $\ell$ grows, for any state $s$ and $0\leq\scalarPrice <\scalarPrice'$ fixed.}
\smallskip

This is not a natural, or immediately intuitive assumption.
However, we will show that it is satisfied if the valuation of jobs follows a log-concave distribution which is parametrized by the job's length, and where the valuation is (informally) positively correlated with this length.
Log-concave distributions are also commonly referred to as distributions possessing a {\em monotone hazard rate}, 
and it is common practice in economic settings to require this property of the agent valuations. 

\begin{lemma}\label{lem:log-concave-assumption}
Let, $V_\ell^s$ denote the marginal r.v. $V$ conditioned on $L=\ell$ and $D\geq s$.
    Let $Z$ be a continuously-supported random variable, and
    $\gamma_1^s\leq \gamma_2^s\leq \dotsm\in \R$.
    If $V_\ell^s$ is distributed like $\gamma_\ell^s \cdot Z$, 
    $\floor{\gamma_\ell^s \cdot Z}$, $Z+\gamma_\ell^s$, or $\floor{Z+\gamma_\ell^s}$, 
    then Assumption 1  is satisfied if $Z$ is log-concave, 
    or if the $\gamma$'s are independent of $\ell$.
\end{lemma}
A discussion of log-concave random variables and a proof of this fact is given in Appendix~\ref{app:log-concave}. 
Many standard (discrete) distributions are (discrete) log-concave random variables, including the uniform, Gaussian, logistic, exponential, Poisson, binomial, etc.
These can be proved to be log-concave from the discussion in Appendix~\ref{app:log-concave}.
In the above, the $\gamma$ terms represent a notion of spread or shifting, parametrized by the length,
indicating some amount of positive correlation.

It remains to show price monotonicity under the above assumption.
First, we begin with the following,
which holds for arbitrary distributions.
\begin{lemma}
\label{lem:Monot-State-U}
    Let $\Uopt t s$ be the expected future revenue earned starting at time $t$ in state $s$, for the optimal policy computed by Algorithm~\ref{alg:finite_horizon}.
    Then the function $s\mapsto\Uopt t s$ is monotone non-increasing in $s$ for any $t$ fixed.
\end{lemma}
See Appendix~\ref{app:proofs} for the proof.
This lemma
ensures that over-selling time on the server can only hurt the mechanism.
This allows us to conclude

\begin{lemma} \label{lem:monot-prices} \color{black}
    If the distribution on job parameters satisfies the above assumption, then for all $\ell,s,t$, we have $\pi^*_t(s,\ell)\leq \pi^*_t(s,\ell+1)$.
\end{lemma}
\begin{proof}[Sketch.]
A full proof may be found in Appendix~\ref{app:proofs}.
The idea is to show that, for any price $\scalarPrice$ less than the optimum $\pi^*_t(s,\ell)$,
the difference in revenue between charging $\scalarPrice$ and $\pi^*_t(s,\ell)$ to jobs of length $\ell$ is less than the difference in revenue between the same prices for jobs of length $\ell+1$.
This is achieved by applying the assumption to recursive definition of future revenue, along with the previous lemma.
Thus, we can conclude that the optimal price $\pi_t^*(s,\ell+1)$ must be greater than $\pi^*_t(s,\ell)$.
 \end{proof}
With Lemma~\ref{lem:monot-prices} and the
results of Appendix~\ref{app:extensions}, we finally have: 
\begin{theorem}
 \label{thm:monotone-opt} The online server pricing problem admits an optimal monotone pricing strategy when the variables $L$, $V$, and $D$ satisfy assumption 1. Also, 
 \begin{enumerate}
     \item In the finite horizon setting, when $\V$ is finitely supported, an exact optimum can be computed in time $O(T\KK^3)$.
     \item In the infinite horizon setting, when $\V$ is finitely supported, for all $\varepsilon>0$, an $\varepsilon$-additive-approximate policy can be computed in time \[ O\left(\KK^3 \log_\gamma\left(\tfrac{\varepsilon(1-\gamma)}{\VV}\right)\right) \leq O\left(\tfrac{\KK^3}{1-\gamma}\ln\left(\tfrac{\VV}{\varepsilon(1-\gamma)}\right)\right)\]
    \item In the finite horizon setting, when $V$ is continuously supported, for all $\eta>0$, an $\eta T$-additive-approximate policy can be computed in time $O(T\KK^2\VV/\eta)$.
 \end{enumerate}
\end{theorem}

\subsection{Concentration Bounds on Revenue for Online Scheduling}\label{sec:concentration-full-knowledge}
\label{subsec:concentration-all-policies}
In this section, we show that the revenue of arbitrary policies concentrates around their mean. In particular it holds true for the optimal or approximately optimal strategies described above. This will also allow us to argue later that, if we have an estimate $\hat Q$ of $Q$, then execute Algorithm~\ref{alg:finite_horizon} given the distribution $\hat Q$, then the output policy will perform well with respect to $Q$, both in expectation, and with high probability. 
To show this concentration, we will consider the {\em Doob} or {\em exposure} martingale of the cumulative revenue function, introduced in Section~\ref{sec:model}.
Define
\begin{equation}\label{eq:martingale-def}
     R_i^\pi:=\E{}{\crev_T(\pi,\mathcal X)\middle|\mathcal X_1,\,\dotsc,\,\mathcal X_i}
\end{equation}
where the $\mathcal X_i$'s are jobs in the sequence $\mathcal X$ and the expected value is taken with respect to $\X_{i+1}, \dots \X_T$.
Thus, $R_0^\pi$ is the expected cumulative revenue, and $R_T^\pi$ is the random cumulative revenue.
To formally describe this martingale sequence, we will introduce some notation, and formalize some previous notation.
Recall that $\X_1 ,\, \X_2,\, \dotsc$ is a sequence of jobs sampled {\em i.i.d.} from an underlying distribution $Q$.  
Fix a pricing policy ${\pi:[T]\times \SS\times \L\to\R}$. 
Note that the state at time $t$ is a random variable depending on both the (deterministic) pricing policy and the (random) $\X_1, \dots, \X_{t-1}$.
We denote it $S_t(\pi,\mathcal X)$, or $S_t$ for short. 
Formally, suppose $\X_t=(V_t,L_t,D_t)$, then $S_{t+1}(\pi,\X)=S_{t}(\pi,\X)-1$ if either $V_t<\pi_t(S_t,L_t)$ or $D_t<S_t$,
and otherwise $S_{t+1}(\pi,\X)=S_t(\pi,\X)+L_t-1$.
Furthermore, let $\rev_t(\pi,\X)$ be equal to 0 in the first case above (the $t$-th job is not scheduled), and  $\pi_t(S_t,L_t)$ otherwise. 
Thus, $S_t(\pi,\X)$ and $\rev_t(\pi,\X)$ are functions of the random values $\X_1,\,\dotsc,\,X_t$ for $\pi$ fixed. Note that $\rev_t$ implicitly depends on $S_t$.
Let $\X_{>i}:=(\X_{i+1},\X_{i+2}, \ldots)$ and $\X_{\leq i}:=(\X_1, \ldots \X_i)$.
Recalling that $\crev_T(\X,\pi)=\sum_{t=1}^T\rev_t(\X,\pi)$, we have
\begin{subequations}\label{eq:definition-of-revenue-martingale}
\begin{align}
    R_i^\pi &=  \sum_{t=0}^i \rev_t(\pi,\X) + \E{\X_{>i}}{\textstyle\sum_{t=i+1}^T
    \rev_t(\pi,\X)\;\middle|\;S_{i+1}(\pi,\X_{\leq i})}\\
    &= \left(\textstyle\sum_{t=0}^i \rev_t(\pi,\X_{\leq t})\right) + \Upi {i+1}{S_{i+1}(\pi,\X_{\leq i})}
\end{align}\end{subequations}
We wish to show that $\crev(\X,\pi)$ concentrates around its mean.
Since $R_0^\pi$ is the expected revenue due to $\pi$, and $R_T^\pi$ is the (random) revenue observed, 
it suffices to show $|R_0^\pi-R_T^\pi|$ is small,
which we will do by applying Azuma's inequality,
after showing the bounded-differences property.
This gives, see Appendix \ref{app:concentration} for details,
\begin{theorem}
\label{thm:concentration-finite-horizon-opt}
	Let $\X$ be a finite sequence of $T$ jobs sampled $i.i.d.$ from $Q$,
	and let $\pi$ be any monotone policy. Then,
	with probability $1-\delta$, \\[-8pt]
	\[
		\left|\crev_T(\X,\pi) - \E{\X'}{\crev_T(\X',\pi)}\right|
		\leq \VV\cdot \sqrt{2\log(\sfrac 2\delta)T}.
	\]
	in the finite horizon, and in the infinite-horizon-discounted,\\[-8pt]
	\[
		\left|\crev_\infty(\X,\pi) - \E{\X'}{\crev_T(\X',\pi)}\right|
		\leq \VV\cdot \sqrt{2\log(\sfrac 2\delta)/(1-\gamma^2)}.
	\]
	In particular these results hold true for the (approximately) optimal pricing strategies of Theorem \ref{thm:monotone-opt}. 
\end{theorem}

\section{Robustness of Pricing with Approximate Distributional Knowledge}\label{sec:policy-with-learning}
In this section,  we show that results analogous to Theorems \ref{thm:monotone-opt} and \ref{thm:concentration-finite-horizon-opt} may be obtained even in the case in which we do not have full knowledge of the distribution $Q$, but only an estimate $\hat Q$. We then show how to obtain a valid $\hat Q$ from samples.

\subsection{Robustness of the pricing strategy}
\label{sec:robustness}
Let's suppose that instead of knowing the exact distribution $Q=(D,L,V)$ of the jobs, we have only access to some estimate $\hat Q=( \hat D, \hat L, \hat V)$ with the following property, for some $\varepsilon >0 $:
\begin{equation}
    \label{eq:robustness}
    \left| \mathbb{P}(\hat L = \ell,\hat V\geq v, \hat D\geq s ) - \mathbb{P}(L=\ell, V\geq v, D\geq s)\right| <  \varepsilon \ \ \forall s \in S, \ell \in \L \text{ and }v \in \V.
\end{equation}
For the sake of brevity, we abuse notation and denote the condition in $(\ref{eq:robustness})$ as $|Q - \hat Q| < \varepsilon.$
Later, we will need to estimate the value $\mathbb{P}[L=\ell, \neg(V\geq v, D\geq s)]$, given $\hat Q$, that is the probability that the job has length $\ell$, but either cannot afford price $v$, or cannot be scheduled $s$ slots in the future.
This is equal to $
    \mathbb{P}[L=\ell]-\mathbb{P}[L=\ell,V\geq v,D\geq s]\ $.

The left-hand term is equal to $\mathbb{P}[{L=\ell,}{V\geq0,}{D\geq 0]}$, 
and so we have access to both terms.
The estimation error is additive, so the deviation is at most~$2\varepsilon$.\\
Denote $\bm p^\ell_{t,s}:=\mathbb{P}[V\geq \pi^t(s,\ell),D\geq s|L=\ell]$, and recall
\begin{equation}\Upi t s :=\sum_{\ell\in\L}\mathbb{P}[L=\ell]\Big(\bm p^\ell_{t,s}\big(\pi_t(s,\ell)+\Upi{t+1}{s+\ell-1}\big)+(1-\bm p^\ell_{t,s})\Upi{t+1}{s-1}\Big),
\label{eq:U-def}
\end{equation}
the expected revenue from time $t$ onwards, conditioning on $S_t=s$.
Let $\hUpi t \cdot$ be the same as $\Upi t \cdot$, but where the variables are distributed as $\hat Q$. 
As before, let $\Uopt t \cdot$ be $\Upi t \cdot$ for $\pi=\pi^*$, the Bayes-optimal policy returned by Algorithm~\ref{alg:finite_horizon}, and $\hUopt t \cdot$ defined similarly 
but with respect to $\hat Q$.
We will show that $\hUopt t \cdot$ is a good estimate for $\Uopt t \cdot$.
\begin{lemma}\label{lem:U-error}
    Let $Q$, and $\hat Q$ such that $|Q - \hat Q| < \varepsilon$.
    \begin{enumerate}
        \item In the finite horizon, $|\Uopt t s - \hUopt t s|<2(T-t)\mathbb V \mathbb L \varepsilon$ for all $t,s$;
        \item In the infinite horizon, $|\Uopt {} s - \hUopt {} s|<2\mathbb V \mathbb L \varepsilon / (1-\gamma)$ for all $s$, where $U^*$ is the optimal time independent strategy.
    \end{enumerate}
\end{lemma}
The proof of 1 is in Appendix \ref{app:robustness}, and the proof of 2 in Appendix \ref{sec:Infinite-Horizon-Approx}. 

\subsection{Learning the Underlying Distribution from Samples}\label{subsec:learning-hat-Q}
As discussed above, we show here how to compute a $\hat Q$ from samples of $Q$, such that $|Q-\hat Q|$ is small with high probability. 
In particular we present a sampling procedure which respects the rules of the pricing server mechanism. When a job arrives, we only learn its length, and only if it agrees to be scheduled.
Thus, we are not given full samples of $Q$, complicating the learning procedure. 
Thanks to the previous section, we know that a policy which is optimal with respect to $\hat Q$ will be close-to-optimal with respect to $Q$.\\
We remark, however, that the power of the results of the previous section is not exhausted by this application: one may apply directly the robustness results to specific problems in which the $\hat Q$ is subject to (small) noise or an approximate distribution is already known from other sources. 
\\
Let $\mathcal X=\{(L_1,V_1,D_1),\,\dotsc,$ $\,(L_n,V_n,D_n),\}$ be an {\em i.i.d.} sample of $n$ jobs from the underlying distribution $Q$. Note that the expectation of an indicator is the probability of the indicated event. Fix a length $\ell$, a state $s$, and a value $v$. 
As a consequence of H\"offding's inequality, with probability $1-\delta$, 
\begin{equation}
\Big|	\tfrac 1n\textstyle\sum_{k=1}^n  \mathbb I[L_k=\ell, V_k\geq v, D_k\geq s]
	-\mathbb{P}[L=\ell, V\geq v, D\geq s]\Big| \leq \sqrt{\tfrac{\log(\sfrac2\delta)}{2n}}\label{eq:sample-estimator}
\end{equation}

\subparagraph*{Sampling Procedure.}
We wish to estimate the value $\mathbb{P}[L=\ell,V\geq v,D\geq s]$ for all choices of $\ell$, $v$, and $s$.
Fixing $v$ and $s$, we may repeatedly post prices $\pi_t(s,\ell)=v$ and declare that the earliest available time is $s$, then record 
(i)~which job accepts to be scheduled, 
and (ii)~the length of each scheduled job.
Let $\perr>0$ and $n\geq \log(2/\delta) / (2\perr^2)$, then by~\eqref{eq:sample-estimator}, the sample-average of each value will have error at most $\perr$ with probability $1-\delta$, for any one choice of $(\ell,v,s)$.

Repeating this process for all $\leq \KK^2$ choices of $v\in\V$ and $s\in\SS$ gives us estimates for each. 
Now, if we want to have the estimate hold over all choices of $\ell,v,s$, it suffices to take the union bound over all $ \leq \KK^3$ values (incl. $\ell$), and scaling accordingly.
If we take $n\geq 3\log(2\KK/\delta)/(2\perr^2)$ samples for each of the $\leq \KK^2$ choices of $v$ and $s$, 
then simultaneously for all $\ell$, $v$, and $s$, the
quantity in~\eqref{eq:sample-estimator} is at most $\perr$. So we have obtained the ``$|Q-\hat Q|<\perr$'' condition. It should be noted that, for this sampling procedure, if a job of length $\ell$ is scheduled, we must possibly wait at most $\ell$ times units before taking the next sample to clear the buffer. This blows up the sampling {\em time} by a factor of $O(\LL)$. The following result follows immediately from Lemma \ref{lem:U-error} and H\"offding's inequality for the right choice of $n$.  
\begin{lemma}\label{lem:U-error-samples}
    Let $n$, $Q$, and $\hat Q$,
    be as above.
	In the finite horizon, for all $\varepsilon>0$, if $n\geq 6T\KK^4\log(2\KK/\delta)/\varepsilon^2$, we have that with probability $1-\delta$, $|\Uopt t s - \hUopt t s|<\varepsilon$ for all $t,s$.
	In the infinite horizon, if $n\geq 6\KK^4\log(2\KK/\delta)/((1-\gamma)\varepsilon^2)$, we have that with probability $1-\delta$, $|\Uopt {} s - \hUopt {} s|<\varepsilon$ for all $s$.
\end{lemma}

\subsection{Performance of the Computed Policy}
\label{sec:performance}
We use here the result of the previous sections to analyze the performance of the policy output by Algorithm~\ref{alg:finite_horizon} after the learning procedure. 
By the estimation of revenue, the best policy in estimated-expectation is near-optimal in expectation. 
Since revenues from arbitrary policies concentrate, we get near-optimal revenue in hindsight.\\
Formally, for $\varepsilon>0$, Lemma~\ref{lem:U-error-samples} gives us that if the sample-distribution $\hat Q$ is computed on $n\geq 6T\KK^4\log(2\KK/\delta)/\varepsilon^2$ samples, then with probability $1-\delta$ over the samples, $|\Uopt t s - \hUopt t s|\leq\varepsilon$. 
Note that $\Uopt{t=0}{s=0}$ is exactly the expected cumulative revenue of the optimal policy.
For clarity of notation, denote
\begin{equation}
    \ecrev_T(\pi|Q):= \E{\X\sim Q}{\crev_T(\X,\pi)}
\end{equation}
We have shown that for sufficient samples, $|\ecrev_T(\pi^*|Q)-\ecrev_T(\pi^*|\hat Q)|<\varepsilon$, with probability $1-\delta$.
This observation allows us to then conclude

\begin{theorem}[Finite Horizon]
\label{thm:summary}
    Let $Q$ be the underlying distribution over jobs.
    Let $\varepsilon>0$, and $n\geq 24T\KK^4\log(8\KK/\delta)/\varepsilon^2$.
    Then in time $O(T\KK^3+n\LL)$, we may compute a policy $\hat \pi$ which is monotone in length, and therefore incentive compatible, such that for any policy $\pi$, with probability $(1-\delta)$,
    \[
        \crev_T(\X,\hat \pi)\geq \crev_T(\X,\pi) - 2\VV\sqrt{2\log(\sfrac 8\delta)(T+1)} - \varepsilon
    \]
    Furthermore, if the distribution over values $V$ is continuous rather than discrete, we may compute in time $O(T\KK^2\VV/\eta+n\LL)$ a monotone policy $\hat \pi$ such that for any policy $\pi$, with probability $1-\delta$,
    \[
        \crev_T(\X,\hat \pi)\geq \crev_T(\X,\pi) - 2\VV\sqrt{2\log(\sfrac 8\delta)(T+1)} - \varepsilon - \eta T
    \]
\end{theorem}

\begin{proof}
    We have chosen $n\geq 6T\KK^4\log(2\KK/(\delta/4))/(\varepsilon/2)^2$.
    Let $\pi^*$ be the optimal policy for the true distribution $Q$. 
    By Theorem~\ref{thm:concentration-finite-horizon-opt}, we have $|\crev_T(\X,\pi)-\ecrev_T(\pi|Q)|<\VV\sqrt{2\log(8/\delta)(T+1)}$ with probability $1-\delta/4$ for both $\pi$ and $\hat \pi$. 
    Furthermore, by Lemma~\ref{lem:U-error-samples}, $|\ecrev_T(\pi|Q)-\ecrev_T(\pi|\hat Q)|<\varepsilon/2$ with probability $1-\delta/4$, for both $\pi=\hat \pi$ and $\pi^*$. This is because from the point of view of $\hat \pi$, $\hat Q$ is the true distribution, and $Q$ is the estimate.
    Taking the union bound over all four events above, and recalling that $\hat \pi$ maximizes $\ecrev_T(\pi|\hat Q)$, and $\pi^*$ maximizes $\ecrev_T(\pi|Q)$, we get the following with probability $1-\delta$:
    \begin{align*}
        \crev_T(\X,\hat \pi)&\geq 
        \ecrev_T(\hat \pi|Q)-\VV\sqrt{2\log(8/\delta)(T+1)}&\text{(concentration)}\\
&\geq \ecrev_T(\hat \pi|\hat Q)-\VV\sqrt{2\log(8/\delta)(T+1)}-\varepsilon/2\!\!\!\!&\text{(sample error)}\\
&\geq \ecrev_T(\pi^*|\hat Q)-\VV\sqrt{2\log(8/\delta)(T+1)}-\varepsilon/2\!\!\!\!&\text{(optimality)}\\
&\geq \ecrev_T(\pi^*|Q)-\VV\sqrt{2\log(8/\delta)(T+1)}-\varepsilon\!\!\!\!&\text{(sample error)}\\
        &\geq \ecrev_T(\pi|Q)-\VV\sqrt{2\log(8/\delta)(T+1)}-\varepsilon\!\!\!\!&\text{(optimality)}\\
        &\geq \crev_T(\X,\pi)-2\VV\sqrt{2\log(8/\delta)(T+1)}-\varepsilon\!\!\!\!&\text{(concentration)}
    \end{align*}
    as desired.\\
    When $V$ is continuously distributed, choose prices which are multiples of $\eta$ between 0 and $\VV$, as is outlined in Appendix~\ref{sec:Approx-Grid-Values}.
 \end{proof}

For what concerns the $\gamma$-discounted infinite horizon case, we have the following

\begin{theorem}[Infinite Horizon, Discounted]
\label{thm:summary-infinite}
    Let $Q$ be the underlying distribution over jobs.
    Let $\varepsilon>0$, and $n\geq 24\KK^4\tfrac{\log(8\KK/\delta)}{\varepsilon^2(1-\gamma)}$.
    Then we may compute a policy $\hat \pi$ in time $O\left(\tfrac{\KK^3}{1-\gamma}\ln\left(\tfrac{\VV}{\varepsilon(1-\gamma)}\right)+n\LL\right)$, which is monotone, and thus incentive compatible, such that for any policy $\pi$, with probability $(1-\delta)$,
    \[
        \crev_\infty(\X,\hat \pi)\geq \crev_\infty(\X,\pi) - 2\VV\sqrt{2\log(\sfrac 8\delta)/(1-\gamma^2)} - 2\varepsilon
    \]
    Furthermore, if the distribution over values $V$ is continuous rather than discrete, we may compute in time $O\left(\tfrac{\KK^2\VV/\eta}{1-\gamma}\ln\left(\tfrac{\VV}{\varepsilon(1-\gamma)}\right)+n\LL\right)$ a monotone policy $\hat \pi$ such that for any $\pi$, with probability $1-\delta$,
    \[
        \crev_\infty(\X,\hat \pi)\geq \crev_\infty(\X,\pi) - 2\VV\sqrt{2\log(\sfrac 8\delta)/(1-\gamma^2)} - 2\varepsilon - \eta/(1-\gamma)
    \]
\end{theorem}
As above, this policy $\hat \pi$ is computed by learning $\hat Q$ from $n$ samples as in Section~\ref{subsec:learning-hat-Q}, then running the modified Algorithm~\ref{alg:finite_horizon} for the estimated distribution as in Appendix \ref{sec:Infinite-Horizon-Approx}. In case $V$ is continuously distributed, we restrict ourselves to prices which are multiples of $\eta$ between 0 and $\VV$. The details of the proof are in Appendix \ref{app:extensions}.\\
We recall that all these results need the distribution assumption from Section~\ref{sec:Monotone-Opt}.

 \section{Conclusions and Future Work}
 \label{sec:conclusion}

In summary, we propose to price time on a server by first learning the distribution over jobs from samples, then computing the Bayes-optimal policy from the estimated distribution.
Our learning algorithm is simple: we sample the distribution through the observation of $n$ jobs at artificially fixed prices and server-states, and learn the job parameters depending on whether they accept to be scheduled.
Using these observations, we build an observed distribution $\hat{Q}$.  
 We then run Algorithm \ref{alg:finite_horizon} with $\hat Q$ and compute an optimal policy $\hat{\pi}$ for $\hat{Q}$. We are guaranteed that the policy prices monotonically (due to Lemma~\ref{lem:Monot-State-U}), and therefore it is incentive compatible, which implies the correctness of the estimated revenue.

 \subparagraph*{Future Work.} There are many natural extensions to this work. For example, one could consider a multi-server setting, settings where jobs can request to be scheduled later than the earliest available time, or settings where jobs need various quantities of differing resources, such as memory and computation time.

\bibliographystyle{unsrt}  

\bibliography{references}

\appendix

\section{Log-Concave Distributions}\label{app:log-concave}
In Section~\ref{sec:Monotone-Opt}, we sought to show that if the value of a random job has a log-concave distribution, then the optimal policy will be monotone. 
We present here a discussion of log-concavity, both for continuous and discrete random variables, and give the proof of the monotonicity of the prices. 

Formally, a function $f:\R\to\R$ is log-concave if for any $x$ and $y$, and for any $0\leq \theta\leq 1$, $\lg f(\theta x + (1-\theta)y)\geq \theta \lg  f(x) + (1-\theta)\lg f(y)$. Equivalently, $f(\theta x + (1-\theta)y)\geq f(x)^\theta f(y)^{1-\theta}$. 
For a discretely supported $f:\mathbb Z\to \R$, we replace this condition with $f(x)^2 \geq f(x-1)f(x+1)$, emulating the continuous definition with $\theta=\tfrac12$.
We further require that the support of $f$ be ``connected''.

\begin{definition}
    A continuous random variable $X$ with density function $f$ is said to be log-concave if $f$ is log-concave. 
    A discrete random variable $Y$ with probability mass function $p$ is said to be log-concave if $p$ is discretely log-concave.
\end{definition}
A well-known fact is that log-concave random variables also have log-concave cumulative density/mass functions. 
We present here a quick proof of this fact, for completeness. 
\begin{claim}
    If $X$ is a log-concave continuous r.v., then $\mathbb{P}[X\leq x]$, and ${\mathbb{P}[X\geq x]}$ are log-concave functions of $x$.
    If $Y$ is a log-concave discrete r.v. supported on~$\mathbb N$, then $\mathbb{P}[Y\leq y]$ and $\mathbb{P}[Y\geq y]$ are discretely log-concave functions of $y$.
\end{claim}
\begin{proof}
    The continuous case is well-documented in the literature.
    See for example~\cite{bagnoli2005log}.
    For the discrete case,
    observe first that since a mass function is non-negative, 
    and we have assumed contiguous support, 
    the function must be single-peaked, {\em i.e.} quasi-concave, as any local minimum would contradict the definition.
    Furthermore, the definition of log-concavity is equivalent to
    \(
        \tfrac {p_y}{p_{y-1}}\geq \tfrac{p_{y+1}}{p_y}
    \).
    Repeatedly applying this, and rearranging, we get
    \[
        p_yp_{y+k} \geq p_{y-1}p_{y+k+1}\quad \forall y,k\in \mathbb Z,\, k\geq 0\ .
    \]
    
    It remains to show that $P(y):=\sum_{-\infty}^y p_k$ is log-concave. We have
    \begin{align*}
        P(y)P(y)
        &= P(y-1)P(y)+ \sum_{-\infty}^y p_kp_y\\
        &\geq P(y-1)P(y) + \sum_{-\infty}^y p_{k-1}p_{y+1}
        = P(y-1)P(y+1)
    \end{align*}
    as desired. The same technique applies for the upper-sum.
 \end{proof}

This will allow us to then conclude:
\subparagraph*{(Lemma~\ref{lem:log-concave-assumption}, p.\pageref{lem:log-concave-assumption})}
    \textit{Let, $V_\ell^s$ denote the marginal r.v. $V$ conditioned on $L=\ell$ and $D\geq s$.
    Let $Z$ be a continuously-supported random variable, and
    $\gamma_1^s\leq \gamma_2^s\leq \dotsm\in \R$.
    If $V_\ell^s$ is distributed like $\gamma_\ell^s \cdot Z$, 
    $\floor{\gamma_\ell^s \cdot Z}$, $Z+\gamma_\ell^s$, or $\floor{Z+\gamma_\ell^s}$, 
    then Assumption 1  is satisfied if $Z$ is log-concave, 
    or if the $\gamma$'s are independent of $\ell$.}
\begin{proof}
First, observe that 
\[
    \mathbb{P}[V\geq\scalarPrice,D\geq s|L=\ell] = 
    \mathbb{P}[V\geq\scalarPrice|D\geq s,L=\ell]\cdot\mathbb{P}[D\geq s|L=\ell]\ .
\]
and since we are taking ratios for $s$ fixed, we can replace the joint cumulatives on $V$ and $D$ in the assumption, with the marginals on just $V$.

Now, if the $\gamma$'s are independent of $\ell$, then the ratio remains unchanged as $\ell$ changes, satisfying assumption 1.
Otherwise, we begin by analyzing the distributions given by $\gamma_\ell^s Z$ and $Z+\gamma_\ell^s$. 
Let $\F(x):= \mathbb{P}[Z\geq x]$, noting that $\mathbb{P}[V_\ell^s\geq \scalarPrice] = \F(\scalarPrice/\gamma_\ell^s)$ and $\F(\scalarPrice - \gamma_\ell^s)$, for the two cases, respectively.
Note that we wish to show $\mathbb{P}[V_\ell^s\geq \scalarPrice']/\mathbb{P}[V_\ell^s\geq \scalarPrice]$ is increasing, which is equivalent to $\log(\mathbb{P}[V_\ell^s\geq \scalarPrice'])-\log(\mathbb{P}[V_\ell^s\geq \scalarPrice])$ increasing.

For $V_\ell^s\sim Z+\gamma_\ell^s$, observe that for $x'>x$ and $\gamma'>\gamma$, we have
\[
    \log \F(x-\gamma) - \log\F(x'-\gamma) \geq \log \F(x-\gamma') - \log \F(x'-\gamma')
\]
since $\log \F$ is a non-increasing and concave function, by assumption.
Also
\begin{align*}
    \log \F(x/\gamma) - \log\F(x'/\gamma) &\geq \log\F(x/\gamma') - \log\F(x/\gamma' + (x'-x)/\gamma)\\
    &\geq \log\F(x/\gamma') - \log\F(x'/\gamma')
\end{align*}
where the first inequality is the same as the previous equation, as the second is by monotonicity.
Thus we have done the continuous case.

For $V_\ell^s\sim \floor{Z+\gamma_\ell^s}$, we note that $\floor{Z+\gamma}\geq x$ if $Z+\gamma\geq \ceil{x}$.
So the probability is $\F(\ceil x - \gamma)$.
Similarly, for $V_\ell^s\sim \floor{\gamma_\ell^s\cdot Z}$, $\mathbb{P}\floor{\gamma Z}\geq x$ is $\F(\ceil x/\gamma)$.
Thus, if we assume that $x$ and $x'$ are integers, the calculations above go through, 
as desired.
 \end{proof}

We present a final fact that justifies the use of $\floor Z$-type random variables:
\begin{lemma}
    If $Y$ is a discrete log-concave random variable, then there exists a continuous log-concave $Z$ such that $Y\sim \floor Z$.
\end{lemma}
\begin{proof}
    Let $P:\mathbb Z\to [0,1]$ be the right-hand cumulative mass function for $Y$. 
    Then, it suffices to have $\mathbb{P}[Z\geq n] = P(n)$ for all integers $n$. 
    Let $\phi:\R\to\R$ be the piecewise-linear function such that $\phi(-\infty)\to 0$, $\phi(\infty) \to -\infty$, 
    and $\phi(n) = \log(P(n))$ for all $n$. 
    Since $\log(P)$ is a discretely concave and non-increasing function, $\phi$ must be concave and non-increasing. 
    We can then set $Z$ to be the random variable whose density is given by $-\tfrac {\mathrm d}{\mathrm d x} \exp(\phi(x))$.
 \end{proof}

\section{Detailed Proofs}\label{app:proofs}

We present in this section the detailed proofs of the lemmas and theorems from the text. 
\ref{app:proofs:algs} gives the pseudocode for the dynamic programs that compute the optimal pricing policies, outlined in Section~\ref{sec:Bayes-opt-policies},
\ref{app:proofs:monotone-prices} gives the proofs for the monotonicity of the pricing policies, along with the discussion on log-concave random variables from Appendix~\ref{app:log-concave}, ~\ref{app:concentration} gives the concentration bounds from the last part of Section~\ref{sec:Bayes-opt-policies}, and ~\ref{app:robustness} deals with the proof of Section~\ref{sec:policy-with-learning} .

\clearpage \vspace*{4em}

\subsection{MDP Algorithms and Correctness}\label{app:proofs:algs}

\begin{algorithm}
\DontPrintSemicolon
\KwData{MDP with states $\mathcal S$, actions $\mathcal A$, and rewards $R$; and a horizon $T$.}
\KwResult{Optimal policy $\pi^*:[T]\times\mathcal S\to \mathcal A$.}
\Begin{
Initialize $\Uopt T s\gets 0$ for all $s\in\mathcal S$.\;
\For{$t$ from $T-1$ to 0, descending}{
	\For{$s\in \mathcal S$}{
		$\Uopt t s\gets\max_{a\in \mathcal A}\big\{
			\sum_{s'\in \mathcal S}\mathbb{P}[s_{t+1}=s'|s,a]\big(
				\text{Reward}(s\to s'|a)+\Uopt{t+1}{s'}
			\big)
		\big\}$\;
		$\pi^*(t,s)\gets\operatorname{arg\,max}_{a\in \mathcal A}\big\{
			\sum_{s'\in \mathcal S}\mathbb{P}[s_{t+1}=s'|s,a]\big(
				\text{Reward}(s\to s'|a)+\Uopt{t+1}{s'}
			\big)
		\big\}$\;
	}
}
\KwRet{$\pi$}
}
\caption{Backwards Induction for Finite-Horizon MDP's \cite{puterman2005markov}, section 4.5}\label{alg:MDPAlgorithm}
\end{algorithm}

\begin{algorithm}
\DontPrintSemicolon
\KwData{Distribution $Q$, $\LL$, $\VV$, $\mathcal S$ and horizon $T$.}
\KwResult{Optimal policy $\pi^*:[T]\times \mathcal S \times \LL \to \mathcal R$.}
\vspace*{3pt}
Initialize $\Uopt T s \gets 0$ for all $s\in\mathcal S$, and $\Ulen T s \ell \gets 0$ for all $s \in \mathcal S, \ \ell \in \LL$.\;
\For{$t$ from $T-1$ to 0, descending}{
	\For{$s\in \mathcal S$}{
	    \For{$\ell \in \LL$}{\small
	        $\Ulen{t}{s}{\ell} \gets \max_{\scalarPrice \in \VV} \Big\{\mathbb{P}[V \geq \scalarPrice, D \geq s|L=\ell]\ \times $\;
	        \hfill$\big(\scalarPrice + \Uopt{t+1}{s+\ell-1} - \Uopt{t+1}{s-1}\big) +{\Uopt{t+1}{s-1}}
		\Big\}$ \;
		$\pi_t^*(s,\ell)\gets\operatorname{arg\,max}_{\scalarPrice \in \VV} \Big\{\mathbb{P}[V \geq \scalarPrice, D \geq s|L=\ell]\ \times$\;
	        \hfill$\big(\scalarPrice + \Uopt{t+1}{s+\ell-1} - \Uopt{t+1}{s-1}\big) +{\Uopt{t+1}{s-1}}
		\Big\}$\;
	}
	$\Uopt t s \gets \sum_{\ell \in \LL} \mathbb{P}[L=\ell] \Ulen t s \ell.\;$
	}
}
\KwRet{$ \pi^*$}

\caption{Optimal policy in finite horizon}
\label{alg:finite_horizon}
\end{algorithm}

\subparagraph*{(Lemma~\ref{lemma_equivalence}, p.\pageref{lemma_equivalence})}
\textit{For any fixed pricing policy $\pi:[T]\times\SS\times\L\to \R$,\vspace*{-0.35em}}
\[
\Upi t s=\E{L}{\Ulenpi t s L}, \forall t \in T , \ s  \in \SS,\\[-0.35em]
\] \textit{where the $\Upi t {\cdot}$'s are as in \eqref{eq:expected_gain}, and the $\Ulenpi t {\cdot}{\cdot}$'s are from the modified MDP. }

\begin{proof}
The statement is true for $t=T$ since in that case everything is zero.
Suppose $\E{L'}{\Ulenpi{t+1}s{L'}}=\Upi{t+1}s$ for all $s$.
For the fixed policy $\pi$, we define $\bm p^\ell_{t,s}:=\mathbb{P}[V\geq \pi_t(s,\ell), D\geq s|L=\ell]$. Then,
	\begin{align*}
		\E{L}{\Ulenpi t s L}&=\sum_{\ell\in\L}\mathbb{P}[L=\ell]\Ulenpi t s \ell\\
		&=\sum_{\ell\in\L}\mathbb{P}[L=\ell]\Big(
		\pi_t(s,\ell)\bm p^\ell_{t,s} + 
		\bm p^\ell_{t,s} \E{L'}{\Ulenpi{t+1}{s+\ell-1}{L'}}\\
		&\qquad\qquad\qquad\qquad\qquad\ \ \ + (1-\bm p^\ell_{t,s})\E{L'}{\Ulenpi{t+1}{s-1}{L'}}\Big)\\
		&=\sum_{\ell\in\L}\mathbb{P}[L=\ell]\Big(\pi_t(s,\ell)\bm p^\ell_{t,s} + 
		\bm p^\ell_{t,s} \Ulenpi{t+1}{s+\ell-1}{L'}\\
		&\qquad\qquad\qquad\qquad\qquad\qquad\qquad+ (1-\bm p^\ell_{t,s}){\Ulenpi{t+1}{s-1}{L'}}\Big)\\
		&= \E{\X}{\rev_t(\pi,\X) + \Upi{t+1}{S_{t+1}(S_t,\X)}\,|\,S_t=s}\\
	&	=: \Upi t s\qquad\qquad\qquad\qquad\qquad\qquad\qquad\qquad\qquad\qquad\square
	\end{align*}
\end{proof}

\subsection{Monotonicity of Prices}\label{app:proofs:monotone-prices}
These proofs are given in parallel with the discussion in Appendix~\ref{app:log-concave}.
\subparagraph*{(Lemma \ref{lem:Monot-State-U}, p.\pageref{lem:Monot-State-U})}
    \textit{Let $\Uopt t s$ be the expected future revenue earned starting at time $t$ in state $s$, for the optimal policy computed by Algorithm~\ref{alg:finite_horizon}.
    Then the function $s\mapsto\Uopt t s$ is monotone non-increasing in $s$ for any $t$ fixed.
}
\begin{proof}
The proof is by induction on the time, decreasing.
    At time $t=T$, there is no future revenue and $\Uopt T s=0$, so the inductive claim follows trivially.
    Suppose, now, that the inductive claim holds at time $t+1$.
    It suffices to show that this holds for each $\Ulen t s \ell$, since $\Uopt t s$ is simply their expectation.
    Let $\pi^*_{t}$ be the optimal pricing policy computed for the time $t$ by the Algorithm~\ref{alg:finite_horizon}.
    Since the function $\scalarPrice\mapsto \mathbb{P}[V\geq \scalarPrice \ \text{and}\  \mathcal E]$, for any event $\mathcal E$,
    is left-continuous in the variable $\scalarPrice$, we may define, for every $\ell \in \L$ and $s \in \SS,$
    \[
        \scalarPrice_s':= \max\big\{ \scalarPrice:
            \mathbb{P}[V\geq \scalarPrice,D\geq s|L=\ell]\,\geq\,\mathbb{P}[V\geq \pi^*_t(s+1,\ell), D\geq s+1|L=\ell]
        \big\}
    \]
    We must have $\scalarPrice'\geq \pi^*_t(s+1,\ell)$, 
    as $\mu=\pi^*_t(s+1,\ell)$ is in the set.
    Now, letting $\bm p:= \mathbb{P}[V\geq \pi^*_t(s+1,\ell), D\geq s+1|L=\ell]$, we have
    \begin{align*}
        &\Ulen t {s+1} \ell\\ &= \bm p\cdot \pi^*_t(s+1,\ell)  + 
		\bm p\cdot {\Uopt{t+1}{s+\ell}}+ (1-\bm p){\Uopt{t+1} s}\\
		&\leq \bm p\cdot \pi^*_t(s+1,\ell)  + 
		\bm p\cdot {\Uopt{t+1}{s+\ell-1}}+ (1-\bm p){\Uopt{t+1}{s-1}}&\mkern-54mu\text{\small (by induction)}\\
		&\leq \bm p\cdot \Big( \pi^*_t(s+1,\ell) + \Uopt{t+1}{s+\ell-1} - \Uopt{t+1}{s-1}\Big)_+ + \Uopt{t+1}{s-1}\mkern-24mu\\
		&\leq \mathbb{P}[V\geq \scalarPrice_s', D\geq s]\cdot\Big(\scalarPrice_s' + \Uopt{t+1}{s+\ell-1} - \Uopt{t+1}{s-1}
		\Big)_+ + \Uopt{t+1}{s-1}\mkern-36mu\\
		&\leq \Ulen t s \ell &\mkern-54mu\text{\small (subopt. price),}
    \end{align*}
    where $(x)_+:= \max\{x,0\}$. 
    The first inequality holds by the induction hypothesis, the second is by definition of $(\cdot )_+$, the third by the definition of $\scalarPrice_s'$,
and in the last, from the fact that $\scalarPrice_s'$ is a (possibly) suboptimal pricing policy for the state $s$ at time $t$. Note that this last inequality requires that the 0 value be feasible in the max, which it is, by setting $\mu'$ arbitrarily large.
 \end{proof}

\subparagraph*{(Lemma \ref{lem:monot-prices}, p.\pageref{lem:monot-prices})}
    \textit{If the distribution on job parameters satisfies assumption~1, then for all $\ell,s,t$, we have $\pi^*_t(s,\ell)\leq \pi^*_t(s,\ell+1)$.}

\begin{proof}
    Let $\bm p^\ell_s(\scalarPrice):= \mathbb{P}[V\geq \scalarPrice, D\geq s|L=\ell]$, 
    fix $s$, $t$, and $\ell$, 
    and let $\scalarPrice_0$ be equal to the optimal price $\pi^*_t(s,\ell)$.
    Observe that $\scalarPrice_0$ maximizes the expression
    \[
        \bm p^{\ell}_s(\scalarPrice)\big(
            \scalarPrice + \Uopt {t+1}{s+\ell-1} - \Uopt {t+1}{s-1}
        \big) + \Uopt {t+1}{s-1}
    \]
    For simplicity, let $\Delta_\ell:= \Uopt {t+1}{s+\ell-1} - \Uopt {t+1}{s-1}$, 
    and so for any $\scalarPrice\neq \scalarPrice_0$,
    \begin{align*}
        0&\leq \bm p^{\ell}_s(\scalarPrice_0)\big(
            \scalarPrice_0 + \Delta_\ell
        \big) - \bm p^{\ell}_s(\scalarPrice)\big(
            \scalarPrice + \Delta_\ell
        \big)\\
        &= \Big(\bm p^{\ell}_s(\scalarPrice_0) - \bm p^{\ell}_s(\scalarPrice)\Big)\big(
                \scalarPrice_0 + \Delta_\ell\big) + \bm p^{\ell}_s(\scalarPrice)\big(\scalarPrice_0 - \scalarPrice\big)
    \end{align*}
    Note that, as discussed in the proof of the previous lemma, $\mu_0+ \Delta_\ell\geq 0$, as otherwise it would be beneficial to set $\pi_t^*(s,\ell)\gets \infty$. 
    The above inequality is then equivalent to
    \[
        \frac{\bm p^{\ell}_s(\scalarPrice) - \bm p^{\ell}_s(\scalarPrice_0)}{\bm p^{\ell}_s(\scalarPrice)} \leq 
        \frac{\scalarPrice_0-\scalarPrice}{\scalarPrice_0 + \Delta_\ell} 
        \quad \iff \quad 
        \frac{\bm p^{\ell}_s(\scalarPrice_0)}{\bm p^{\ell}_s(\scalarPrice)} \geq 1- 
        \frac{\scalarPrice_0-\scalarPrice}{\scalarPrice_0 + \Delta_\ell} 
    \]
    We wish to show that, if $\scalarPrice\leq \scalarPrice_0$, then as $\ell$ increases, the above inequality still holds. 
    This would imply that the price $\scalarPrice_0=:\pi_9^*(s,\ell)$ gives better return than $\scalarPrice$ for jobs of length $\ell+1$,
    implying that the optimal price must be at least $\pi_t^*(s,\ell)$, which is our desired goal. 
    
    Now, by assumption~1, the left-hand-side is non-decreasing in $\ell$, 
    so it remains to show that the right-hand-side is non-increasing in $\ell$. 
    The only changing term is $\Delta_\ell$, which by Lemma~\ref{lem:Monot-State-U}, is non-increasing in $\ell$.
    Since it is in the denominator of a subtracted, non-negative term, we have our desired result.
 \end{proof}

\subsection{Concentration Bounds on Revenue for Online Scheduling}
\label{app:concentration}

\subparagraph*{(Theorem \ref{thm:concentration-finite-horizon-opt}, p.~\pageref{thm:concentration-finite-horizon-opt})}
\textit{Let $\X$ be a finite sequence of jobs sampled $i.i.d.$ from an underlying distribution $Q$,
	and let $\pi$ be any monotone policy. Then,
	with probability $1-\delta$,} \\[-8pt]
	\[
		\left|\crev_T(\X,\pi) - \E{\X'}{\crev_T(\X',\pi)}\right|
		\leq \VV\cdot \sqrt{2\log(\sfrac 2\delta)T}.
	\]
	\textit{in the finite horizon, and in the infinite-horizon-discounted,}\\[-8pt]
	\[
		\left|\crev_\infty(\X,\pi) - \E{\X'}{\crev_T(\X',\pi)}\right|
		\leq \VV\cdot \sqrt{2\log(\sfrac 2\delta)/(1-\gamma^2)}.
	\]
	\textit{In particular these results hold true for the (approximately) optimal pricing strategy computed in the previous part of the section.}  
\begin{proof}
	For the finite horizon, we apply Azuma's inequality to the martingale $R_t^\pi$. 
We being by showing the bounded-differences property.
	Note that we do not require truthful behaviour from the jobs, since taking strategic behaviour into account for a non-monotone policy is equivalent to modifying the distribution over the jobs, and making the distribution state-dependent, by increasing the length of those jobs who would rather buy a longer interval.
Thus,
\begin{align*}
   &\left| R^\pi_{t+1}-R^\pi_{t}\right|\\&\quad =\Big|\textstyle\sum_{\tau=0}^{t+1} \rev_\tau(\pi,\X) + \E{\X_{>{t+1}}}{\textstyle\sum_{\tau=t+2}^T
    \rev_\tau(\pi,\X)\;\middle|\;S_{t+2}(\pi,\X_{\leq t+1})}\\
    &\quad\qquad\qquad -\textstyle\sum_{\tau=0}^{t} \rev_\tau(\pi,\X) - \E{\X_{>{t}}}{\textstyle\sum_{\tau=t+1}^T
    \rev_\tau(\pi,\X)\;\middle|\;S_{t+1}(\pi,\X_{\leq t})}\Big|\\
    &\quad=\left| \rev_{t+1}(\pi,\X) - \mathbb{E}_{\X_{t+1}}[\rev_{t+1}(\pi,\X)|S_{t+1}(\pi,\X_{\leq t})]\right|\leq \VV
\end{align*}
where the last inequality follows from properties of conditional expectation.
	With this property, we can apply Azuma's, and get
	\begin{align*}
		\left|\crev_T(\X,\pi) - \E{\X'}{\crev_T(\X',\pi)}\right|
		&\leq \sqrt{2\log(\sfrac 2\delta)(T+1)\VV^2}.
	\end{align*}
    For the infinite-horizon-discounted, we 
    observe that equation~\eqref{eq:definition-of-revenue-martingale} becomes 
    \begin{equation*}
    R_i^\pi =  \sum_{t=0}^i \gamma^t\rev_t(\pi,\X) + \E{\X_{>i}}{\textstyle\sum_{t=i+1}^T
    \gamma^t\rev_t(\pi,\X)\;\middle|\;S_{i+1}(\pi,\X_{\leq i})}
    \end{equation*}
and thus we get that $|R_t^\pi-R_{t-1}^\pi|\leq \gamma^t\VV$. Therefore with probability $1-\delta$, 
\begin{align*}
    |R_T^\pi-R_0^\pi|&\leq \sqrt{2\log(2/\delta)\textstyle\sum_{t=0}^T(\gamma^t\VV)^2}\ =\ \VV \cdot \sqrt{2\log(2/\delta)\textstyle\sum_{t=0}^T(\gamma^2)^t}
\end{align*}
Thus, taking the limit as $T\to \infty$, we get that with probability $1-\delta$, \[
		\left|\crev_\infty(\X,\pi) - \E{\X'}{\crev_T(\X',\pi)}\right|
		\leq \VV\cdot \sqrt{2\log(\sfrac 2\delta)/(1-\gamma^2)}.\\[-2em]
	\]
 \end{proof}

\subsection{Robustness of the pricing strategy}\label{app:robustness}

\subparagraph*{(Lemma~\ref{lem:U-error}, p.~\pageref{lem:U-error})} \textit{Let $Q$, and $\hat Q$ such that $|Q - \hat Q| < \varepsilon$. In the finite horizon, $|\Uopt t s - \hUopt t s|<2(T-t)\mathbb V \mathbb L \varepsilon$ for all $t,s$.}
\begin{proof}
    Let $\pi^*$ be the policy computed by Algorithm~\ref{alg:finite_horizon} with access to $Q$.
    As in Section~\ref{sec:Bayes-opt-policies}, we denote $\bm p^\ell_{t,s}:= \mathbb{P}[V\geq \pi^*_t(s,\ell),D\geq s|L=\ell]$, and ${\bm P}(\ell):=\mathbb{P}_{\X}[L=\ell]$.
    In an abuse of notation, denote $\hp$ and $\hP$ the estimated values of $\bm p^\ell_{t,s}$ and ${\bm P}(\ell)$, respectively. 
    We cannot estimate $\bm p^\ell_{t,s}$ directly with good error bounds, but we will only need the values $\hP\hp$ and $\hP(1-\hp)$.
	Now, substituting these estimates into~\eqref{eq:U-def}, we get:
	\begin{subequations}
		\begin{align}
			&\left|\Uopt t s-\hUopt t s\right|\nonumber
\\&\quad = \bigg|\sum_{\ell\in\L}{\bm P}(\ell)\Big(\bm p^\ell_{t,s}\pi^*_t(s,\ell) + 
\bm p^\ell_{t,s} \Uopt{t+1}{s+\ell-1} + (1-\bm p^\ell_{t,s})\Uopt{t+1}{s-1}\Big)\nonumber\\
    &\qquad\quad - \sum_{\ell\in\L}\hat{\bm P}(\ell)\Big(\hp\pi^*_t(s,\ell) + 
    \hp \hUopt{t+1}{s+\ell-1} + (1-\hp)\hUopt{t+1}{s-1}\Big)\bigg|\nonumber
		\end{align}\end{subequations}
	To simplify this expression, we begin by showing a simple claim: Let $x$, $y$, $\hat x$, $\hat y\in \R$, and let $\lambda,\hat\lambda \in [0,1]$, such that $|x-\hat x|<\delta$, $|y-\hat y|<\delta$, and $|\lambda-\hat\lambda|<\varepsilon$. Then
	\begin{align*}
	    \big|\big(\lambda x +& (1-\lambda) y\big)-\big(\hat\lambda \hat x + (1-\hat \lambda) \hat y\big)\big|\\&\leq \big|\big(\lambda x + (1-\lambda) y\big)-\big(\hat \lambda x + (1-\hat \lambda) y\big)\big|+ \\ & \qquad \qquad 
	    \big|\big(\hat \lambda x + (1-\hat \lambda) y\big)-\big(\hat \lambda \hat x + (1-\hat \lambda) \hat y\big)\big|\\
	    &\leq |\lambda-\hat\lambda|\cdot|x-y| + \hat\lambda|x-\hat x| + (1-\hat\lambda)|y-\hat y|\\
	    &\leq \varepsilon|x-y|+\delta
	\end{align*}
	
	Now, replacing $x$ and $y$ with $\big(\pi^*_t(s,\ell)+\Uopt{t+1}{s+\ell-1}\big)$ and $\Uopt{t+1}{s-1}$, respectively, and replacing $\lambda$ with $\bm P(\ell)\bm p^\ell_{t,s}$, we have
	\begin{align}
			\left|\Uopt t s-\hUopt t s\right|
&\leq \sum_{\ell\in\L}\bigg( 2\perr\cdot\sup_{\sigma}\left|\pi_t^*(\sigma,\ell)+\Uopt{t+1}{\sigma+\ell-1}-\Uopt{t+1}{\sigma-1}\right| +\nonumber\\&\qquad\qquad + \hP\cdot\sup_{\sigma'}\left|\Uopt{t+1}{\sigma'}-\hUopt{t+1}{\sigma'}\right|\bigg)\label{eq:sample-error-intermediate}
		\end{align}
		However, the argument of the supremum in left-hand term in the summand must be at most $\VV$, since if $\Uopt{t+1}{\sigma+\ell-1}\leq \Uopt{t+1}{s-1}$, it is best to $\pi^*_t(\sigma)=\infty$, which makes $\bm p^\ell_{t,s}=0$, putting all the weight on $\Uopt{t+1}{s-1}$. Furthermore, we have shown in Lemma~\ref{lem:Monot-State-U} that $\Uopt{t+1}{s+\ell-1}\leq \Uopt{t+1}{s-1}$.
 		Thus, we get
	\begin{align*}
			\left|\Uopt t s-\hUopt t s\right|
&\leq \textstyle\sup_{\sigma'}\left|\Uopt{t+1}{\sigma'}-\hUopt{t+1}{\sigma'}\right| + \textstyle\sum_{\ell\in\L} 2\perr\VV
		\end{align*}
	
	Inductively applying this gives $\left|\Uopt t s-\hUopt t s\right|\leq 2(T-t)\LL\VV\perr$ as desired.
 \end{proof}

\section{Extensions}\label{app:extensions}

In this section,
we extend the finite-horizon results to compute the optimal policies in the infinite-horizon-discounted setting, and also to argue that the optimal policy may be computed within some error when the distribution over values is continuous, rather than discrete.

These results are needed to show the full statements of Theorems~\ref{thm:monotone-opt}--\ref{thm:summary-infinite}.

\subsection{\textbf{Infinite Discounted Horizon}}\label{sec:Infinite-Horizon-Approx}

Recall, in this infinite horizon discounted setting, we seek to maximize the $\gamma$-discounted future revenue, 
\[
    \crev_{\infty}(\mathcal X,\pi):=\sum_{t=0}^{\infty} \gamma^t\rev_t(\mathcal X,\pi)
\]
over the choice of $\pi:\mathbb N\times \SS\times \L\to \R$.
Algorithm~\ref{alg:finite_horizon} does not allow us to immediately compute a solution for the infinite discounted horizon case.
However we can exploit the discounting factor on the revenues to obtain an approximation of the infinite optimum: it suffices to consider the truncated problem up to a certain sufficiently large  $T$ and solve it optimally using the algorithm presented above. In fact we have the following Lemma.  

\begin{lemma} \label{lemma:infinite_horizon}
For any $\varepsilon>0$ and $T\geq \log_\gamma\left(\varepsilon(1-\gamma)/\VV\right)$, 
let $\pi$ be the pricing policy computed by the finite-horizon algorithm up to time $T$. Let $\bar\pi$ be the time-independent pricing policy such that $\bar\pi(\cdot,\cdot):=\pi_0(\cdot, \cdot)$.
Then the expected performance of the optimal policy in the infinite horizon is within an additive $\varepsilon$ of expected performance of $\bar\pi$.\end{lemma}
\begin{proof}
    Note that in order to compute policy $\pi$ it is necessary to add the discount factor to Algorithm~\ref{alg:finite_horizon}, and to all of the proofs of previous sections. One can verify that all proofs go through.
    Let $\pi^*$ be the Bayes-optimal infinite-horizon strategy --- which is known to be time-independent --- and let $\pi$ be as in the statement (where we set $\pi_t(s,\ell)=\infty$ for all $t>T$.)
    Then, in expectation over times $0$ through $T$, pricing as $\pi$ yields greater revenue than following $\pi^*$.
    Conversely, in expectation over all time, pricing as $\pi^*$ yields greater revenue than $\pi$.
    However, after time $T$, the maximum possible revenue due to any policy is 
    \[
        \textstyle\sum_{t=T}^\infty \gamma^t\cdot \VV = \gamma^T \cdot \VV \cdot \left(\tfrac 1{1-\gamma}\right) \ \leq\ \varepsilon
    \]
    And so the difference in revenue due to following $\pi$ or $\pi^*$ is at most $\varepsilon$, since $T$ is sufficiently large.
    
    It remains to show that $\bar \pi$ performs better than $\pi$ overall.
    Let $\pi^i$ be the policy which agrees with $\pi_0$ for all $t\leq i$, 
    then equals $\pi_{t-i}$ for $t>i$.
    Observe that, $\pi^1$ is optimal in expectation over the interval $[1,T+1]$, and is equivalent to $\pi=\pi^0$ for the first step. Therefore, $\pi^1$ performs better than $\pi$.
    Similarly, we can argue $\pi^{i+1}$ performs better than $\pi^{i}$ over the interval $[i,T+i]$ and equally before, hence performs better overall.
    
    Thus, we have a sequence of policies $\pi=\pi^0,\,\pi^1,\,\pi^2,\,\dotsc$ converging to $\bar \pi$, and whose expected revenue is monotone non-decreasing along the sequence. Therefore, the expected revenue due to $\bar \pi$ is greater than that of $\pi$, which is an $\varepsilon$ additive-approximation to the optimal policy.
 \end{proof}
The approach above is analogous to the classical value iteration technique \cite{puterman2005markov}.

\subparagraph*{(Lemma~\ref{lem:U-error}, p.~\pageref{lem:U-error})} \textit{Let $Q$, and $\hat Q$ such that $|Q - \hat Q| < \varepsilon$. In the infinite horizon, $|\Uopt {} s - \hUopt {} s|<2\mathbb V \mathbb L\varepsilon / (1-\gamma)$ for all $s$.}
\begin{proof}
    As in the proof of Lemma~\ref{lemma:infinite_horizon}, if $T$ is sufficiently large, we may analyze the first $T$ time steps as a finite-horizon problem, and the remaining revenue will be negligibly small. 
    Now, the calculation above can be reproduced with discount terms, to show
    \begin{align*}
			\left|\Uopt t s-\hUopt t s\right|
&\leq \textstyle\sup_{\sigma'}\left|\gamma\Uopt{t+1}{\sigma'}-\gamma\hUopt{t+1}{\sigma'}\right| + \textstyle\sum_{\ell\in\L} 2\perr\VV
		\end{align*}
    Then, inductively applying this and taking $T\to \infty$, we have $
        |\Uopt{0}{s}-\hUopt{0}{s}|\leq 2\VV\LL\perr/(1-\gamma)$.
 \end{proof}

These results are used to prove the infinite-horizon versions of the various results throughout the paper, specifically the Theorems~\ref{thm:monotone-opt}--\ref{thm:concentration-finite-horizon-opt},~and~\ref{thm:summary-infinite}.

\subsection{\textbf{Approximation Algorithm for Continuously Supported Values}} \label{sec:Approx-Grid-Values}
Note that the algorithms above assume that the {\em value} of the jobs ($V$) is discretely supported, and the running time depends on~$|\V|$. In this section, we analyze the error incurred by discretizing the space of possible values, and then computing the optimal policy.

Let $\eta>0$ be some desired small grid size, and suppose we only allow ourselves to set prices which are multiples of $\eta$. We claim that this incurs a small loss on the total revenue.

Define, as in the previous subsections, $\bm p^\ell_{s}(\scalarPrice):=\mathbb{P}[V\geq\scalarPrice,D\geq s|L=\ell]$.
Further, define as previously $\Uopt t s=\E{L}{\Ulen t s L}$, and \[
    \Ulen t s \ell := \max_{\scalarPrice\in
    \R}\left[\bm p^\ell_s(\scalarPrice)\big(\scalarPrice+\Uopt{t+1}{s+\ell-1}\big) + (1-\bm p^\ell_s(\scalarPrice))\Uopt{t+1}{s-1}\right]
\]
Define $\Ugrid t s$ and $\Ugridlen t s \ell$ similarly, restricting the maximum to choosing $\scalarPrice$ from multiples of $\eta$.

\begin{lemma}\label{lem:revenue-loss-if-grid}
    Let $\Uopt t \cdot$ and $\Ugrid t \cdot$ be defined as above, then $|\Uopt t s-\Ugrid t s|\leq (T-t)\eta\ \forall s,t$.
\end{lemma}
\begin{proof}
    We will show this by induction on the value of $t$, decreasing. 
    Assume that $|\Uopt t s-\Ugrid t s|<\Delta_t$ for all $t$, $s$, and set $\Delta_{T+1}=0$. We wish to inductively bound the value of $\Delta_t$. Now,
    \begin{align*}
    &\Ugridlen t s \ell \\&= \max_{\scalarPrice\in
    \eta\cdot\mathbb Z}\left[\bm p^\ell_s(\scalarPrice)\big(\scalarPrice+\Ugrid{t+1}{s+\ell-1}\big) + (1-\bm p^\ell_s(\scalarPrice))\Ugrid{t+1}{s-1}\right]\\
    &\geq\max_{\scalarPrice\in
    \eta\cdot\mathbb Z}\Big[\bm p^\ell_s(\scalarPrice)\big(\scalarPrice+\Uopt{t+1}{s+\ell-1}-\Delta_{t+1}\big) +\\&\mkern236mu +(1-\bm p^\ell_s(\scalarPrice))\Uopt{t+1}{s-1}-\Delta_{t+1}\Big]\\
    &=-\Delta_{t+1}+\max_{\scalarPrice\in
    \eta\cdot\mathbb Z}\left[\bm p^\ell_s(\scalarPrice)\big(\scalarPrice+\Uopt{t+1}{s+\ell-1}-\Uopt{t+1}{s-1}\big) + \Uopt{t+1}{s-1}\right]
    \end{align*}
    Now, let $\scalarPrice^*$ be the optimizer of this right hand side over $\R$ (where the value would attain $\Ulen t s \ell$), and $\hat \scalarPrice$ be $\scalarPrice^*$ rounded {\em down} to the nearest multiple of $\eta$.
    Then, since $\bm p^\ell_s(\cdot)$ is non-increasing,
    \begin{align*}
        \bm p^\ell_s(&\hat \scalarPrice)\big(\hat \scalarPrice+\Uopt{t+1}{s+\ell-1}-\Uopt{t+1}{s-1}\big) + \Uopt{t+1}{s-1}\\
        &\geq \bm p^\ell_s(\scalarPrice^*)\big(\scalarPrice^*-\eta+\Uopt{t+1}{s+\ell-1}-\Uopt{t+1}{s-1}\big) + \Uopt{t+1}{s-1}\\
        &= \Ulen t s \ell - \eta\cdot \bm p^\ell_s(\scalarPrice^*)
    \end{align*}
    Thus combining both equations, we get 
    \[
        \Ugridlen t s \ell \ \leq \ \Ulen t s \ell \ \leq \ 
        \Ugridlen t s \ell + \eta + \Delta_{t+1}
    \]
    From which we conclude, by averaging over $\ell$, that $\Delta_t\leq (T-t)\eta$, as desired.
 \end{proof}

\begin{corollary}
    Let $\Uopt {} \cdot$ and $\Ugrid {\infty} \cdot$ be defined as above, but for the infinite horizon discounted, then $|\Uopt {} s-\Ugrid {\infty} s|\leq \eta/(1-\gamma)\ \forall s$.
\end{corollary}
\begin{proof}
    As shown in the previous subsection, it suffices to perform the analysis in the finite horizon, while taking the discount factor into account, then take the limit as $T\to \infty$.
    The same calculations as above gives
    \begin{align*}
        &\Ugridlen {t} s \ell \\& \geq -\Delta_{t+1} + \max_{\scalarPrice\in\eta\mathbb Z}\left[
            \bm p^\ell_s(\scalarPrice)\left(\scalarPrice + \gamma\Uopt {t+1}{s+\ell-1} - \gamma\Uopt{t+1}{s-1} \right) + \gamma\Uopt{t+1}{s-1}
        \right]\\&\geq \Ulen{t+1}s \ell - \eta - \gamma\Delta_{t+1}
    \end{align*}
    Summing the $\Delta$'s and taking $T\to\infty$, we get $\Ugridlen{\infty} s \ell \geq \Ulen {} s \ell - \eta/(1-\gamma)$ as desired.
 \end{proof}

\end{document}